\documentclass[10pt, conference, letterpaper]{IEEEtran}
\ifCLASSINFOpdf
\else
\fi
\usepackage[numbers,sort&compress]{natbib}
\usepackage{amsfonts}
\usepackage{amsmath}
\usepackage{amssymb}
\usepackage{subeqnarray}
\usepackage{cases}
\usepackage{array}
\usepackage[ruled,linesnumbered]{algorithm2e}
\usepackage{algorithmic}
\usepackage{bm}
\usepackage{cuted}
\usepackage{multicol}
\usepackage[numbers,sort&compress]{natbib}
\usepackage{extarrows}

\usepackage{multirow}
\usepackage{makecell,diagbox}
\usepackage{balance}
\usepackage{graphicx}
\usepackage{subfigure}
\usepackage{mathrsfs}
\usepackage{ntheorem}
\usepackage{url}
\graphicspath{{pic/}}

\theoremheaderfont{\itshape}
\theorembodyfont{\normalfont\rm}
\newenvironment{proof}{{\noindent\it Proof}.\;}{\hfill $\square$\par}
\newtheorem{definition1}{\hspace{1em}Definition}
\theoremseparator{.}
\theoremheaderfont{\bf}
\newtheorem{theorem}{\hspace{0em}Theorem}
\newtheorem{remark}{\hspace{0em}Remark}

\newcommand{\PreserveBackslash}[1]{\let\temp=\\#1\let\\=\temp}
\newcolumntype{C}[1]{>{\PreserveBackslash\centering}p{#1}}
\newcolumntype{R}[1]{>{\PreserveBackslash\raggedleft}p{#1}}
\newcolumntype{L}[1]{>{\PreserveBackslash\raggedright}p{#1}}

\begin{document}
%
\title{Receiver-driven Video Multicast over NOMA Systems in Heterogeneous Environments}
%
%
%
\author{\IEEEauthorblockN{Xiaoda Jiang$^{\dag}$, Hancheng Lu$^{\dag}$, Chang Wen Chen$^{\ddag}$$^{\S}$, Feng Wu$^{\dag}$}
\IEEEauthorblockA{$^{\dag}$School of Information Science and Technology, University of Science and Technology of China, China\\
$^{\ddag}$School of Science and Engineering, Chinese University of Hong Kong, Shenzhen, China\\
$^{\S}$Department of Computer Science and Engineering, State University of New York at Buffalo, USA\\
jxd95123@mail.ustc.edu.cn, hclu@ustc.edu.cn, chencw@cuhk.edu.cn, fengwu@ustc.edu.cn
}}

\maketitle
\begin{abstract}
Non-orthogonal multiple access (NOMA) has shown potential for scalable multicast of video data. However, one key drawback for NOMA-based video multicast is the limited number of layers allowed by the embedded successive interference cancellation algorithm, failing to meet satisfaction of heterogeneous receivers. We propose a novel receiver-driven superposed video multicast (Supcast) scheme by integrating Softcast, an analog-like transmission scheme, into the NOMA-based system to achieve high bandwidth efficiency as well as gradual decoding quality proportional to channel conditions at receivers. Although Softcast allows gradual performance by directly transmitting power-scaled transformation coefficients of frames, it suffers performance degradation due to discarding coefficients under insufficient bandwidth and its power allocation strategy cannot be directly applied in NOMA due to interference. In Supcast, coefficients are grouped into chunks, which are basic units for power allocation and superposition scheduling. By bisecting chunks into base-layer chunks and enhanced-layer chunks, the joint power allocation and chunk scheduling is formulated as a distortion minimization problem. A two-stage power allocation strategy and a near-optimal low-complexity algorithm for chunk scheduling based on the matching theory are proposed. Simulation results have shown the advantage of Supcast against Softcast as well as the reference scheme in NOMA under various practical scenarios.
\end{abstract}


%
\IEEEpeerreviewmaketitle

\textfloatsep=10pt plus 0pt minus 0pt

\section{Introduction}
Non-orthogonal multiple access (NOMA) has been considered as a promising technology to improve bandwidth efficiency in the 5G systems \cite{dai2018survey}, \cite{ding2017survey}, \cite{islam2017power}, by leveraging superposition coding (SC) and successive interference cancellation (SIC). Considering that video traffic will be dominant in growing mobile traffic \cite{index2016forecast}, it is desirable to exploit NOMA for high-volume video services. For video unicast, it has shown that NOMA can provide improved visual satisfaction compared with orthogonal multiple access (OMA) \cite{jiang2018enabling}. Meanwhile, multicast is an effective way to enhance bandwidth efficiency for high-volume video services. With SC, NOMA has the potential to enable scalable data multicast \cite{Choi2015minimum}, \cite{zhang2018tradeoffs}. In the NOMA-based multicast schemes, data is encoded into base-layer (BL) signal and enhanced-layer (EL) signal, which are transmitted simultaneously through SC. With SIC, near users with strong channel gains can decode both BL and EL signals, while far users with weak channel gains may only decode BL signals.

\begin{figure}[htb]
\centering
\includegraphics[width=5cm]{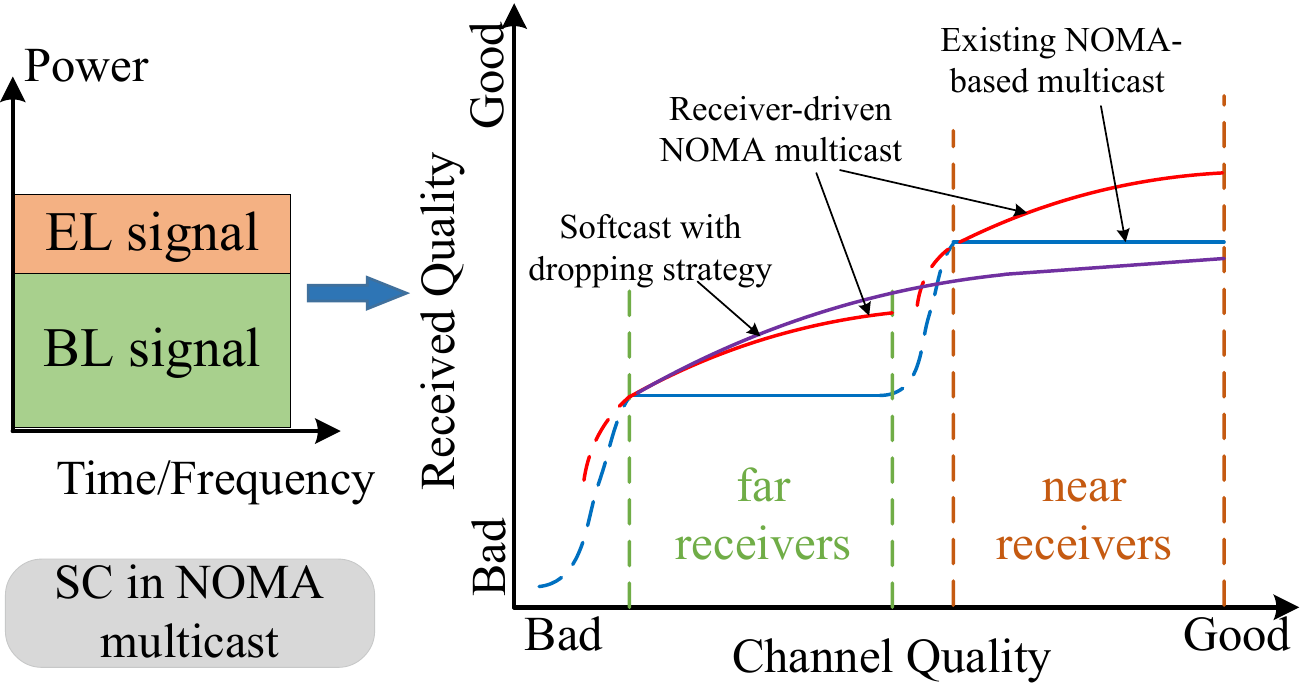}
\vspace{-0.2cm}
\caption{Multcast performance to heterogeneous receivers in NOMA systems under insufficient bandwidth.}\label{figure:cliff}
\vspace{-0.2cm}
\end{figure}
\par
However, these NOMA-based multicast scheme cannot easily meet satisfaction of all receivers with heterogeneous channel conditions. Since SIC implies that all other superposed signals have to be decoded in order before decoding the required signal, the complexity of SIC scales with the number of superposed signals and severe error propagation would occur in incorrect SIC decoding. Thus, existing NOMA-based multicast schemes generally cluster users to allow only two layer signals to be superposed, as Fig. \ref{figure:cliff} shows, the BL signal is pessimistically encoded at a bit rate decided by the decoding ability of the far user with worst channel quality \cite{DVB}. It means other far users cannot receive better video quality proportional to their channel quality. Similar drawback would also occur when encoding the EL signal for near users.

\par
In contrast, the potential to satisfy heterogeneous users has been provided through a recent proposed analog-like video transmission scheme named Softcast \cite{softcast} and its derivative schemes \cite{parcast}, \cite{cui2014robust}, \cite{liu2018cg}, \cite{tan2017analog}, \cite{he2017progressive}. Under Softcast, each receiver can receive video quality proportional to its channel quality. It achieves such robustness by skipping non-linear processes in digital transmission, e.g., quantization, entropy coding and forward error correction. Instead, coefficients, which are linearly transformed from video frames with discrete cosine transformation (DCT), are power-scaled and directly transmitted. Consequently, the perturbation caused by channel noise is linearly added to original video pixels. This enables that the encoded video signal can provided video quality adapted to channel quality of heterogeneous, not discrete quality levels pre-determined by scalable video coding (SVC) at the sender. However, the compression efficiency of Softcast is inferior to digital systems, and it has to discard certain coefficients when bandwidth is inadequate, causing performance degradation due to nonrecoverable distortion of discarding \cite{softcast}, \cite{parcast}, \cite{liang2017hybrid}, as Fig. \ref{figure:cliff} shows.

\par
To address the aforementioned dilemma, we develop a receiver-driven scheme called Supcast ({\bf sup}erposed video multi{\bf cast}) in NOMA systems, attempting to cater for all receivers with heterogeneous channel conditions as well as enhance performance in Softcast in the case of insufficient bandwidth. In Supcast, BL and EL signals are distinguished across DCT chunks, which are generated by grouping nearby DCT coefficients.  Specifically, DCT chunks are bisected into EL chunks and BL chunks, based on chunk characteristics. One EL chunk and one BL chunk compose the superposed signal in a physical packet. Since decoding performance of each chunk is proportional to channel quality, each receiver can obtain satisfied video quality. Owing to SC, more information can be conveyed compared with Softcast when bandwidth is limited.

\par
It should be emphasized that there exist two critical challenges in the design of Supcast. First, existing power allocation principles in Softcast cannot be directly adopted for Supcast due to interference caused by SC. Second, existing NOMA optimization focuses on user scheduling to decide which users are to be superposed \cite{zhai2018energy}, \cite{wu2018optimal}, \cite{fang2017joint}. However, in Supcast, user scheduling has been handled by grouping users requesting the same video contents. In Supcast, DCT chunks are basic units for signal scheduling. With SIC, superposed chunks would be regarded as noise when decoding other chunks. In this case, decoding performance would be determined by assigned superposed chunks, whose allocated power reflects interference strength. Therefore, chunk scheduling, coupled with power allocation, will be the key to ensure the desired video reception quality proportional to channel quality. The solutions to these challenges constitute the main contributions of this paper, which are summarized as follows.

\begin{itemize}
\item In Supcast, we combines the linear video processing of Softcast and the SC operation of NOMA into one framework. By doing so, Supcast can implement receiver-driven video transmission, where received quality is scalable to the heterogeneous channel conditions. Compared with Softcast developed for OMA systems, Supcast can achieve better bandwidth efficiency due to SC.
\item In Supcast, we investigate the joint power allocation and chunk scheduling problem, and formulate it as a distortion minimization problem taking into account the characteristics of video contents in these chunks. This formulated problem is a mixed integer non-linear programming (MINLP) problem, which is an intractable NP-hard problem.
\item To tackle the MINLP problem, we decompose it into two subproblems. For power allocation, a two-stage strategy is developed. For chunk scheduling, we reformulate it as a one-to-one two-sided matching game. EL chunks and BL chunks are viewed as two disjoint player sets, which are matched with each other. A near-optimal and low-complexity matching algorithm is proposed. The stability, convergence, complexity and optimality of the proposed algorithm are analyzed thoroughly.
\end{itemize}
\par
\par
Extensive simulations have been carried out to validate the advantages of the proposed Supcast. The results demonstrate that SupCast outperforms Softcast as well as the reference scheme in NOMA under different scenarios. Considering the complexity of the practical NOMA implementation, only two layers are used in Supcast, the same as that in existing NOMA-based multicast schemes. However, Supcast can be easily extended to one BL and multiple fine-grained ELs, by modeling the chunk scheduling as a one-to-many matching or multi-step two-sided one-to-one matching.

\par
The rest of the paper is organized as follows. In Section \ref{sec:system-description}, we present an overview of the framework of Supcast. In Section \ref{sec:formulation and analysis}, we formulate the joint power allocation and chunk scheduling problem in Supcast using distortion as a metric. Both two-stage power allocation strategy and matching game formulation for chunk scheduling are presented. In Section \ref{sec:algorithm}, we present the matching algorithm with detailed analysis of its stability, convergence, complexity and optimality.  Performance evaluations are presented in Section \ref{sec:evaluation}. In Section \ref{sec:conclusion}, we conclude this paper with a summary.

\begin{figure*}[htb]
\centering
\includegraphics[width=16cm]{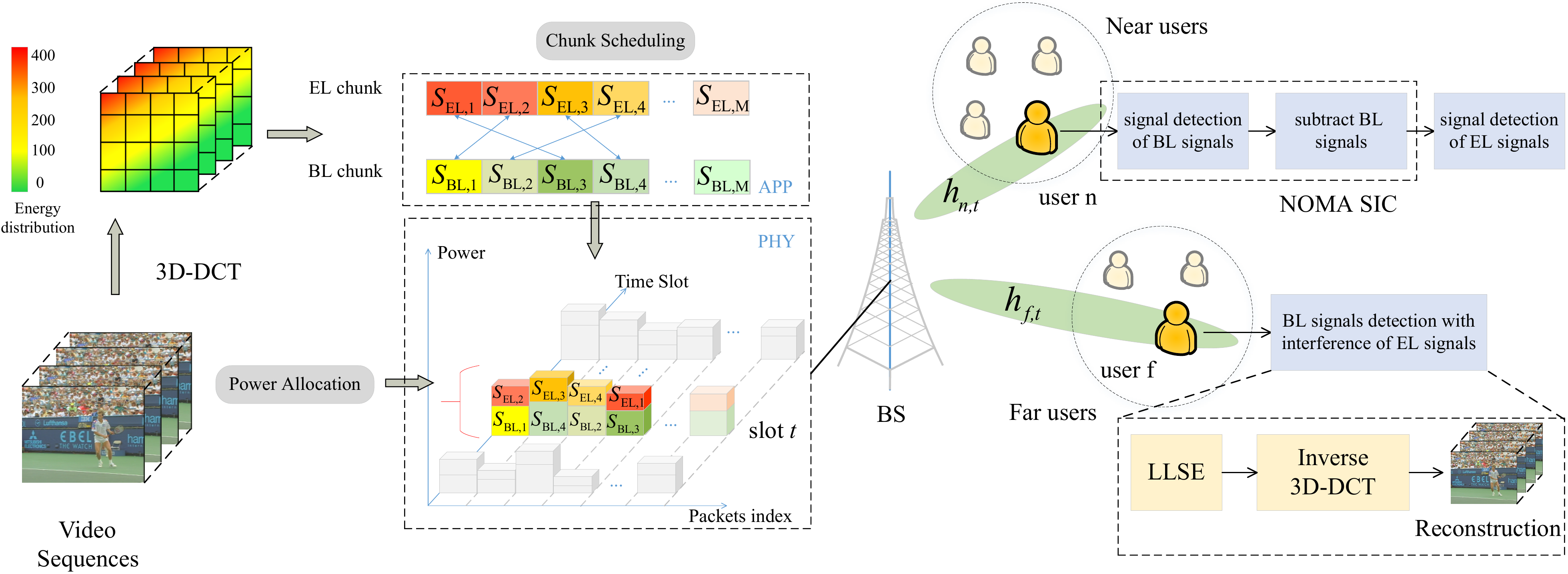}
\vspace{-0.3cm}
\caption{Architecture of Supcast over NOMA networks.}\label{figure:systemOverview}
\vspace{-0.2cm}
\end{figure*}
\section{System Description}\label{sec:system-description}
In this section, we present the framework of the proposed Supcast over the NOMA system, as shown in Fig. \ref{figure:systemOverview}. We shall also elaborate the video encoding process and reconstruction process, which integrate technologies of Softcast and NOMA.
\par
As illustrated in Fig. \ref{figure:systemOverview}, we consider the downlink video transmission from a base station (BS) to multiple users in a single cell. For practical implementation where the complexity should be carefully considered, multicast users are divided into near users and far users according to their distances to the BS, as that in existing NOMA-based multicast schemes. Since Supcast can provide fine-grained reconstruction quality proportional to channel quality, we only need to take one near user denoted by $n$ and one far user denoted by $f$ into consideration. Note that user $n$ and user $f$ have worst channel quality among near and far users, respectively, to ensure that optimized signals can be decoded by other users. Define $h_{n,t}$ and $h_{f,t}$ as the channel gain at the time slot $t$ from the BS to user $n$ and $f$, respectively. Specifically, $h_{n,t}=r_{n,t}/\sqrt{1+d_n^\eta}$ and $h_{f,t}=r_{f,t}/\sqrt{1+d_f^\eta}$, where $d_n$ and $d_f$ are the distance from the BS to user $n$ and user $f$, respectively, $\eta$ is the path-loss exponent, and $r_{n,t},r_{f,t}\sim\mathcal{CN}(0,1)$ are the Rayleigh fading coefficient at the time slot $t$. In addition, we consider that the channel does not vary within a time slot, and assume that the channel state information (CSI) is available at the BS.
\par
The video sequence is encoded into EL signals and BL signals with different transmission priorities. In the downlink NOMA system, each physical packet can simultaneously transmit EL and BL signals with SC. Based on the CSI, the BS assigns different power levels to the superposed signals to distinguish them in the power domain. According to the NOMA principle \cite{dai2018survey}, \cite{Choi2015minimum}, more power would be allocated to BL signals. In this case, user $f$ can only decode the BL signal, regarding the EL signal as interference. For user $n$, the deployed SIC algorithm enables that it can correctly decode the BL signal and subtract it before decoding the EL signal. Therefore, it can decode both BL and EL signals to achieve better reconstruction quality. Generally, the basic unit for video coding is a group of pictures (GOP), which is composed with several successive video frames. Without loss of generality, we assume that transmission of a GOP occupies one time slot. Thus, the processing of video sequences and optimization of transmission are considered with the GOP interval. In the following subsection, we will give a detailed description of how these BL and EL signals are generated and processed.
\vspace{-0.1cm}
\subsection{Softcast-based Video Encoding with SC}\label{subsection:encoding}
In Supcast, video sequences are first encoded based on the linear processing in Softcast. Specifically, a video sequence is divided into GOPs, and then 3D-DCT (three dimensional DCT) is performed over each GOP. The DCT coefficients are divided into equal-sized rectangular-shaped chunks, where coefficients in each chunk are treated as random variables drawn from the zero-mean Gaussian distribution. As Fig. \ref{figure:systemOverview} shows, DCT makes the energy distribution much more compact. Since more energy reflects more importance, we partition each chunk into BL or EL signal according to its variance, which represents the mean energy of contained coefficients. Specifically, the chunks are first sorted in the descend order of variances, and then the sorted chunks are bisected into two sets. $M$ chunks, denoted by $\mathcal{S}_B=\{S_{BL,1},\cdots,S_{BL,M}\}$, form the BL signal with larger variances $[\lambda_{BL,1}, \lambda_{BL,2}, \cdots, \lambda_{BL,M}]$. The other $M$ chunks, denoted by $\mathcal{S}_E=\{S_{EL,1},\cdots,S_{EL,M}\}$, form the EL signal with smaller variances $[\lambda_{EL,1}, \lambda_{EL,2}, \cdots, \lambda_{EL,M}]$.
\par
Without loss of generality, one BL chunk together with one EL chunk is assumed to fit into one physical packet via SC in NOMA. Besides, we assume that the bandwidth of a time slot is enough for transmission of a GOP. If the bandwidth is insufficient, Supcast would take the dropping strategy similar to that in Softcast. Specifically, the least important chunks would be discarded to achieve bandwidth compaction. The remaining chunks are bisected into $M'$ BL chunks and $M'$ EL chunks, where $M'<M$. Before transmission, the BS selects one BL chunk and one EL chunk, e.g., $S_{BL,i}$ and $S_{EL,j}$, for superposition into a NOMA packet, $i,j=\{1\cdots M\}$. This procedure is called chunk scheduling, implemented by the BS according to the available CSI.
\par
Note that each modulated symbol in physical layer contains the I (in-phase) component and Q (quadrature) component. We use $S_{BL,i}^c$ and $S_{EL,j}^c$ to denote the complex source of $S_{BL,i}$ and $S_{EL,j}$, respectively:
\begin{equation}
S_{BL,i}^c=\frac{S_{BL,i}^o+i_z\cdot S_{BL,i}^e}{\sqrt{2}}, ~S_{EL,j}^c=\frac{S_{EL,j}^o+i_z\cdot S_{EL,j}^e}{\sqrt{2}},
\end{equation}
where $S_{BL,i}^o$, $S_{EL,j}^o$ are the odd-index part of $S_{BL,i}$, $S_{EL,j}$, and $S_{BL,i}^e$, $S_{EL,j}^e$ are the even-index part of $S_{BL,i}$, $S_{EL,j}$, respectively. $i_z$ is the imaginary unit.
\par
Before transmission, coefficients in each BL chunk or EL chunk are scaled by the same scaling factor, denoted by $g_{BL,i}$ or $g_{EL,j}$, $\forall i,j=\{1\cdots M\}$. In Supcast, according to the principle of NOMA \cite{dai2018survey}, \cite{Choi2015minimum}, BL chunks should be allocated more power than EL chunks, which guarantees decoding performance of user $f$ with interference of the EL signal. Fortunately, such power allocation principle is coincident with the principle of optimal power allocation in Softcast, which suggests chunks with larger variances should be allocated more power. However, the existed interference and coupled relationship with chunk scheduling, make it unrealistic to directly adopt power allocation results of Softcast into Supcast.
\subsection{Video Reconstruction with SIC and LLSE}\label{subsec:video-reconstruction}
Throughout the paper, the processing of video sequences and optimization are carried out within a GOP duration, corresponding to a time slot as assumed above. Therefore, we omit the time slot index $t$ for brevity in the following paper. Supposing that $S_{BL,i}$ and $S_{EL,j}$ are paired to be conveyed in a packet, the received signals at user $n$ and user $f$ are:
\begin{equation}
\begin{aligned}
Y_{n,i,j} &= h_{n}(g_{BL,i}S_{BL,i}^c+g_{EL,j}S_{EL,j}^c) + W_{n},\\
Y_{f,i,j} &= h_{f}(g_{BL,i}S_{BL,i}^c+g_{EL,j}S_{EL,j}^c) + W_{f},
\end{aligned}
\end{equation}
where $W_{n},W_{f}\sim\mathcal{CN}(0,\sigma_w^2)$ are additive white Gausssian noise (AWGN) for user $n$ and $f$, and $\sigma_w^2$ is the noise variance.
\par
At the receiver side, user $n$ adopts SIC to decode the superposed signals. With perfect SIC, user $n$ can cancel the interference of BL signals \cite{islam2017power}, \cite{Choi2015minimum}, \cite{wu2018optimal}, which are allocated with more power. Specifically, it first correctly decodes the BL signal with SIC, then it subtracts this signal from the received signal and decodes the EL signal without interference. In Supcast, the linear least square estimator (LLSE) is used for the signal decoding. Then inverse 3D-DCT is implemented on the decoded EL chunks to obtain a reconstructed video sequence with superior quality. In this case with perfect SIC, distortion of user $n$ comes from decoding errors of the LLSE of EL chunks. As derived in \cite{softcast}, such distortion in terms of mean square error (MSE) is
\begin{equation}\label{eq:hl-distortion}
d_{n,j} = \frac{\lambda_{EL,j}\sigma_w^2}{|h_{n}|^2g_{EL,j}^2\lambda_{EL,j}+\sigma_w^2}.
\end{equation}
For use $f$, it decodes BL chunks by regarding EL chunks as interference according to NOMA. Thus, based on Eq. (\ref{eq:hl-distortion}), its MSE of decoding is
\begin{equation}\label{eq:ll-distortion}
d_{f,i,j} = \frac{\lambda_{BL,i}(|h_{f}|^2g_{EL,j}^2\lambda_{EL,j}+\sigma_w^2)}{|h_{f}|^2g_{BL,i}^2\lambda_{BL,i}+|h_{f}|^2g_{EL,j}^2\lambda_{EL,j}+\sigma_w^2}.
\end{equation}
Thus, overall distortion of transmitting $S_{BL,i}$ and $S_{EL,j}$ is
\begin{equation}\label{eq:overall-distortion}
D_{i,j} = d_{n,j}+d_{f,i,j}+\sum_{j=1}^M\lambda_{EL,j},
\end{equation}
where $\sum_{j=1}^M\lambda_{EL,j}$ is the MSE distortion caused by user $f$ completely failing to decode EL chunks. Since this part of distortion is constant, it need not be considered in the following optimization.

\vspace{0.5cm}
\section{Problem Formulation and Analysis}\label{sec:formulation and analysis}
In this section, we formulate the power allocation and chunk scheduling problem as a distortion minimization problem. To develop low-complexity and near-optimal solution for this NP-hard problem, we decompose it into two subproblems. For power allocation, we analyze it in two stages with consideration of chunk diversity and interference. For chunk scheduling, we reformulate it by utilizing the matching theory.
\subsection{Problem Statement and Formulation}
The objective of the joint power allocation and chunk scheduling problem is to minimize overall transmission distortion, which is measured by the MSE. Define the binary variable $\mu_{i,j}$ as the indicator for chunk scheduling:
\begin{equation}\label{eq:chunk-indicator}
\mu_{i,j}=
\begin{cases}
    1,&S_{BL,i} {\rm~\, is ~\,paired ~\,with} ~\,S_{EL,j},\\
    0, &{\rm otherwise}.
\end{cases}
\end{equation}
When $\mu_{i,j}=1$, $S_{BL,i}$ and $S_{EL,j}$ are simultaneously transmitted via SC, which brings transmission distortion as expressed in Eq. (\ref{eq:overall-distortion}). Hence, the optimization problem can be formulated as:
\begin{subequations}\label{equation:original-optimization}
\setlength{\abovedisplayskip}{-1pt}
\setlength{\belowdisplayskip}{6pt}
\begin{align}
\!\!\!\!\min_{\bm{\mu,g_{BL},g_{EL}}} \! &\sum_{i=1}^M\sum_{j=1}^M \mu_{i,j}D_{i,j}\label{equation:original-optimization-a}\\
\text{s.t.}~~~\ &\sum_{i=1}^M\sum_{j=1}^M\! \mu_{i,j}(g_{BL,i}^2\lambda_{BL,i}\!+\!g_{EL,j}^2\lambda_{EL,j}) \!\leq\! P^{t},\label{equation:original-optimization-b} \\
&g_{EL,j}^2\lambda_{EL,j} \leq g_{BL,i}^2\lambda_{BL,i}, ~{\rm if}~ \mu_{i,j}=1,\label{equation:original-optimization-c}\\
&\sum_{i=1}^M\mu_{i,j}\!=\!1,\sum_{j=1}^M\mu_{i,j}\!=\!1,~\forall i,j=\{1\cdots M\},\label{equation:original-optimization-d}\\
&\mu_{i,j}\in\{0,1\},~~\forall i,j=\{1\cdots M\}.\label{equation:original-optimization-e}
\end{align}
\end{subequations}
Note that the optimization variable $\bm \mu$ is the chunk scheduling matrix with entries $\mu_{i,j}$. The optimization variables $\bm {g_{BL}}$ and $\bm {g_{EL}}$ are the $M$-dimensional power scaling vectors for BL and EL chunks, respectively. Constraint (\ref{equation:original-optimization-b}) ensures that the total transmitted power for a GOP does not exceed the budget $P^t$. Constraint (\ref{equation:original-optimization-c}) guarantees the SIC decoding according to NOMA power allocation principle. Constraint (\ref{equation:original-optimization-d}) and (\ref{equation:original-optimization-e}) indicate that each BL chunk can only be superposed with one EL chunk and vice verse.
\par
Since the optimization problem in (\ref{equation:original-optimization}) involves both continuous variables $\bm {g_{BL}}$, $\bm {g_{EL}}$ and binary variable $\bm \mu$, it is an NP-hard MINLP problem \cite{di2016joint}, \cite{boyd2004convex}. It is unrealistic to find the global optimal solution. To tackle this coupled problem with low-complexity yet near-optimal solution, we divide it into two subproblems. One is the power allocation problem, which can be handled in two stages. The other is the chunk scheduling problem which can be reformulated as a one-to-one two-sided matching problem.

\subsection{Two-stage Power Allocation}
Given chunk scheduling is determined, that is the two-tuple set for chunk superposition is given as $\mathcal{T}_{c}$. Then problem (\ref{equation:original-optimization}) can be rewritten as
\begin{subequations}\label{equation:power-optimization}
\begin{align}
\!\min_{\bm{g_{BL},g_{EL}}}&\sum_{\{i,j\}\in\mathcal{T}_{c}}D_{i,j}\label{equation:power-optimization-a}\\
\text{s.t.}~~\ &\sum_{\{i,j\}\in\mathcal{T}_{c}}(g_{BL,i}^2\lambda_{BL,i}\!+\!g_{EL,j}^2\lambda_{EL,j}) \!\leq\! P^{t},\label{equation:power-optimization-b} \\
&~~g_{EL,j}^2\lambda_{EL,j} \leq g_{BL,i}^2\lambda_{BL,i}, ~\forall \{i,j\}\in\mathcal{T}_{c}.\label{equation:power-optimization-c}
\end{align}
\end{subequations}
\par
However, due to the existence of interference in the objective function in (\ref{equation:power-optimization-a}), it is not trivial to convert such nonlinear optimization problem into a convex optimization problem.

\par
Actually, the problem in (\ref{equation:power-optimization}) belongs to the class of the sum of generalized polynomial fractional functions (SGPFF) problem. The study in \cite{shen2007global} has shown that global optimal solution can be obtained with a branch and bound algorithm. This algorithm works by solving an equivalent problem, which is further systematically converted into a series of linear programming (LP) problems. However, the number of converted LP problems is related to the dimension of optimization variables. Thus, it is not practical to implement this algorithm to globally optimize the problem in (\ref{equation:power-optimization}), since the dimension of variables is $2M$, which is always large for DCT chunk division. To handle this problem efficiently, we propose a two-stage power allocation strategy with high tractability.
\par
At the first stage, we pre-allocate power across chunks according to their importance to reconstruction, without consideration of the channel gain diversity and interference. In this case, the power pre-allocation problem can be solved by the method of Largrange multiplier. Power allocated to BL and EL chunks at this stage can be derived as
\begin{equation}\label{eq:power-preallocation}
\begin{cases}
P_{BL,i}=\frac{\sqrt{\lambda_{BL,i}}}{\sum_{k\in\mathcal{M}}(\sqrt{\lambda_{BL,k}}+\sqrt{\lambda_{EL,k}})}P^t,\\
P_{EL,j}=\frac{\sqrt{\lambda_{EL,j}}}{\sum_{k\in\mathcal{M}}(\sqrt{\lambda_{BL,k}}+\sqrt{\lambda_{EL,k}})}P^t.
\end{cases}
\end{equation}
The derivation process of such power distortion optimization can be traced from \cite{softcast}. As Eq. (\ref{eq:power-preallocation}) shows, power allocated to the BL (EL) chunk is proportional to its energy $\lambda_{BL,i}$ ($\lambda_{EL,j}$), which reflects the chunk importance.
\par
At the second stage, we re-allocate power within each superposed EL-BL chunk pair, given the power pre-allocation result obtained in the first stage. If $\{i,j\}\in\mathcal{T}_{c}$, the optimization problem for paired $S_{BL,i}$ and $S_{EL,j}$ is formulated as follows.
\begin{subequations}\label{eq:power-reallocation}
\setlength{\abovedisplayskip}{-1pt}
\setlength{\belowdisplayskip}{6pt}
\begin{align}
\!\min_{{g_{BL,i},g_{EL,j}}}&D_{i,j}\label{eq:power-reallocation-a}\\
\text{s.t.}~~\ &g_{BL,i}^2\lambda_{BL,i}\!+\!g_{EL,j}^2\lambda_{EL,j} \!\leq\! P^t_{i,j},\label{eq:power-reallocation-b} \\
&g_{EL,j}^2\lambda_{EL,j} \leq g_{BL,i}^2\lambda_{BL,i},\label{eq:power-reallocation-c}
\end{align}
\end{subequations}
where $P^t_{i,j}=P_{BL,i}+P_{EL,j}$.  Since $g_{BL,i}$ and $g_{EL,i}$ are non-negative, it can be proven that problem (\ref{eq:power-reallocation}) has unique solution. Substituting Eq. (\ref{eq:hl-distortion})-(\ref{eq:overall-distortion}) and Eq. (\ref{eq:power-preallocation}) into the objective (\ref{eq:power-reallocation-a}), we can derive the optimal solution as
\begin{equation}\label{eq:power-reallocation-solution}
\begin{cases}
g_{EL,j}^*\!=\!min([g_{EL,j}^\upsilon]^+,~~\frac{1}{2}\frac{\sqrt{\lambda_{BL,i}}+\sqrt{\lambda_{EL,j}}}{\sum_{k\in\mathcal{M}}(\sqrt{\lambda_{BL,k}}+\sqrt{\lambda_{EL,k}})}P^t),\\
g_{BL,i}^*\!=\!(\!\frac{(\sqrt{\lambda_{BL,i}}+\sqrt{\lambda_{EL,j}})P^t}{\sum\limits_{k\in\mathcal{M}}(\sqrt{\lambda_{BL,k}}+\sqrt{\lambda_{EL,k}})}\!-\!(g_{EL,j}^*)^2\lambda_{EL,j})^{\frac{1}{2}}\frac{1}{\sqrt{\lambda_{BL,i}}},
\end{cases}
\end{equation}
where $g_{EL,j}^\upsilon=(\frac{\sigma_w}{h_nh_f}\sqrt{\frac{h_f^2P^t_{i,j}+\sigma_w^2}{\lambda_{BL,i}\lambda_{EL,j}}}-\frac{\sigma_w^2}{h_n^2\lambda_{EL,j}})^\frac{1}{2}$, and $[x]^+$ means $max(x,0)$. The solution is generated by finding the stationary point. Due to space limits, the detailed derivation is omitted.

\subsection{Two-sided Matching Formulation for Chunk Scheduling}\label{subsec:matching-formulation}
Now we utilize the matching theory to formulate the chunk scheduling problem in (\ref{equation:original-optimization}) under given power allocation. Some definitions and notations are given in the following content.

\subsubsection{Definition}
To describe the mutual relationship between BL chunks and EL chunks, we consider chunk scheduling as a one-to-one two-sided matching process between the set of $M$ BL chunks and the set of $M$ EL chunks. The BS considers these two disjoint sets of chunks as selfish and rational players. Since perfect CSI is available at the BS, these players have complete information of each other when matching. We say $S_{BL,i}$ and $S_{EL,j}$ are matched together and form a {\em matching pair}, if $S_{BL,i}$ and $S_{EL,j}$ are superposed for transmission through a NOMA packet. Based on these, we can formulate the chunk scheduling problem as a typical matching problem, presented as
\begin{definition1}\label{def:matching}
\emph{(One-to-One Two-sided Matching):} Consider BL chunks and EL chunks as two disjoint sets, $\mathcal{S}_B=\{S_{BL,1},\cdots,S_{BL,M}\}$ and $\mathcal{S}_E=\{S_{EL,1},\cdots,S_{EL,M}\}$, respectively. A one-to-one, two-sided {\em matching} $\Phi$ is a mapping from the set of BL chunks $\mathcal{S}_B$ into the EL chunks set $\mathcal{S}_E$, such that for every $S_{BL,i}\in\mathcal{S}_B$ and $S_{EL,j}\in\mathcal{S}_E$ satisfying
\begin{itemize}
\setlength{\itemindent}{0.5em}
  \item [1)]$\Phi(S_{BL,i})\in\mathcal{S}_E$,
  \item [2)]$\Phi(S_{EL,j})\in\mathcal{S}_B$,
  \item [3)]$|\Phi(S_{BL,i})|=1$, $|\Phi(S_{EL,j})|$=1,
  \item [4)]$S_{EL,j}=\Phi(S_{BL,i})\Leftrightarrow S_{BL,i}=\Phi(S_{EL,j})$,
\end{itemize}
\end{definition1}
where $\Phi(S_{BL,i})$ represents $S_{BL,i}$'s partner in $\Phi$ and $\Phi(S_{EL,j})$ represents $S_{EL,j}$'s partner in $\Phi$. Conditions 1), 2) and 3) state that each BL chunk is matched with one EL chunk, and vise verse. Such one-to-one setting is due to the complexity of SIC decoding in NOMA. Condition 4) implies $S_{BL,i}$ and $S_{EL,j}$ are matched with each other.

\subsubsection{Preference Lists}\label{subsubsec:preference}
It should be emphasized that the result of such matching game is greatly influenced by the competition and decision process among players \cite{baron2011peer}. To better characterize these dynamic interactions, each player has own {\em preferences} over the players in the other set. It has been studied that different settings of preferences would have various properties, which may lead to different designs of matching algorithms \cite{roth1992two}, \cite{bayat2014distributed}. In this paper, we take the sum-distortion in Eq. (\ref{eq:overall-distortion}) to directly decide the order of preferences.
\par
The BS can set the preference list for each player, which is ranked in a descending order by the value calculated in Eq. (\ref{eq:overall-distortion}) paired with the player in other set. For example, for any $S_{BL,i}\in\mathcal{S}_B$ and $S_{EL,j},S_{EL,k}\in\mathcal{S}_E$:
\begin{equation}\label{eq:preference}
S_{EL,j} \succ_{S_{BL,i}} S_{EL,k} \Leftrightarrow D_{i,j} < D_{i,k}
\end{equation}
implies that $S_{BL,i}$ prefers $S_{EL,j}$ to $S_{EL,k}$ since the former can provide lower sum-distortion, i.e., higher utility.

\par
With the above matching model and preference lists formulation, we propose an algorithm to solve the formulated matching problem in the next section.
\section{Matching Algorithm for Chunk Scheduling}\label{sec:algorithm}
In this section, we propose a near-optimal algorithm for chunk scheduling by utilizing the matching theory to reduce computational complexity. Furthermore, thorough analysis of the proposed algorithm is provided.

\subsection{Design and Description of Algorithm}
Inspired by the matching theory \cite{roth1992two}, \cite{bayat2014distributed}, \cite{gale1962college}, we propose the BL-EL chunk matching algorithm (BECMA) for chunk scheduling. Considering the competition behaviour as mentioned in Sec. \ref{subsec:matching-formulation}, the basic idea of BECMA is allowing the BL chunk to make a {\em proposal} to an EL chunk selected from its preference list, and the proposed EL chunk has the right to accept or reject the proposal.

\par
Obviously, the conflict would occur when an EL chunk is so ``popular'' that it receives more than one proposal. Since this is a one-to-one matching, an intuitive question would arise that it should accept which proposal and reject others. To answer it, we first introduce the concept of {\em blocking pair} as follows.
\begin{definition1}\label{def:matching}
\emph{(Blocking Pair):} A BL-EL chunk pair $(S_{BL,i},S_{EL,j})$ is a blocking pair in $\Phi$ if it satisfies $S_{EL,j} \succ_{S_{BL,i}} \Phi(S_{BL,i})$ and $S_{BL,i} \succ_{S_{EL,j}} \Phi(S_{EL,j})$, where $\Phi(S_{BL,i})$ represents $S_{BL,i}$'s partner in $\Phi$ and $\Phi(S_{EL,j})$ represents $S_{EL,j}$'s partner in $\Phi$.
\end{definition1}

\par
According to the above definition and Eq. (\ref{eq:preference}), a blocking pair implies higher utility than the original matching pair. Thus, if a matched EL chunk receives another proposal, it will accept the proposing BL chunk only when they can form a blocking pair.

\par
Now we can elaborate the matching process in BECMA, as presented in Algorithm 1. Each BL chunk makes proposals to the EL chunk in order of its preference list. The BL chunk would pause the process if an EL chunk temporarily accepts its proposal, but continue proposing if it is rejected. Meanwhile, for the proposed EL chunk, it will reject the BL chunk if they cannot form a blocking pair, otherwise it will accept the proposal for consideration. This process ends until no BL chunk needs to propose.
\par

\begin{algorithm}
\DontPrintSemicolon
\SetAlgoHangIndent{0pt}
  \caption{BL-EL Chunk Matching Algorithm (BECMA)}
   \label{alg:Optimal-polyblock}
   \KwIn{Set of BL chunks $\mathcal{S}_B$ and set of EL chunks $\mathcal{S}_E$.}
   \KwOut{Stable matching $\Phi$}
   Set up BL chunks' preference lists. \\
   Set up EL chunks' preference lists. \\
   Set up a set of unmatched BL chunks $\mathcal{S}_B^U$ to record BL chunks who have not been paired with any EL chunk.\\
   \While{$\mathcal{S}_B^U$ {\rm is not empty}}
   {$S_{BL,i}$ proposes to its currently most preferred available EL chunk $S_{EL,j}$.\\
   \eIf{$S_{EL,j}$ {\rm already has a partner} $S_{BL,k}$ {\rm and} $(S_{BL,i},S_{EL,j})$ {\rm is not a} blocking pair}
   {
   $S_{EL,j}$ rejects $S_{BL,i}$ and continues holding $S_{BL,k}$.\\ $S_{BL,i}$ removes $S_{EL,j}$ from its preference lists.}
   {
   $S_{EL,j}$ accepts $S_{BL,i}$ and rejects $S_{BL,k}$.\\
   $S_{BL,k}$ removes $S_{EL,j}$ from its preference lists.\\
   $S_{BL,i}$ is removed from $\mathcal{S}_B^U$ and $S_{BL,k}$ is added into $\mathcal{S}_B^U$.
   }
   }
   Output the matching $\Phi$.
\end{algorithm}

\subsection{Analysis of Algorithm}\label{subsec:algorithm-analysis}

\subsubsection{Complexity}
In BECMA, the computational complexity comes from two phases. One is the sorting phase of establishing preference lists, which requires complexity of $\mathcal{O}(2M^2)$. The other is the matching phase, where each BL chunk will propose at most $M$ times. In the worst case, the complexity of matching is $\mathcal{O}(M^2)$. Thus, the complexity of BECMA is $\mathcal{O}(3M^2)$.
\par
For comparison, we analyze two other matching schemes as shown follows.
\begin{itemize}
  \item Optimal Exhaustive Matching: It generates the optimal result by searching all possible $M!$ combinations, where $!$ denotes factorial. Thus, its complexity is $\mathcal{O}(M!)$.
  \item Random Matching: BL chunks and EL chunks are randomly matched with each other. Therefore, the complexity is $\mathcal{O}(M)$.
\end{itemize}
With above analysis, we can see that the complexity of BECMA is much less than the optimal exhaustive scheme. However, as shown in Sec. \ref{subsec:performance-comparison}, the performance degradation of BECMA compared with the optimal scheme is negligible.

\subsubsection{Stability and Convergence}
We first give the definition of stability for matching as below.
\begin{definition1}\label{def:stable}
\emph{(Stable Matching):} A matching $\Phi$ is stable, if there exists no blocking pair $(S_{BL,i},S_{EL,j})$ in $\Phi$.
\end{definition1}
With Definition 2 and 3, we now can prove the stability and convergence of BECMA.

\begin{theorem}
The proposed BECMA converges to a stable matching $\Phi^*$ with limited iterations.
\end{theorem}
\begin{proof}
First we prove the convergence. In each iteration in Algorithm 1 (Lines 5-13), each $S_{BL,i}\in\mathcal{S}_B$ makes a proposal to its currently most preferred EL chunk, which has not rejected it in previous iterations. Since the total number of EL chunks is $M$, each BL chunk cannot make more than $M$ proposals. In other word, the total number of iterations is less than $M$. Besides, the iteration would end until every BL chunk has been matched. Therefore, BECMA can converge to a final matching $\Phi^*$ after a limited number of iterations.
\par
Next we prove that $\Phi^*$ is stable. With the detail in Algorithm 1, in the final matching $\Phi^*$, there is no any $S_{BL,i}\in\mathcal{S}_B$ can find a $S_{EL,j}\in\mathcal{S}_E$ to form a blocking pair. Thus, according to Definition 3, $\Phi^*$ is a stable matching.
\end{proof}
\subsubsection{Optimality}Here, we give a Theorem to state whether an optimal matching can be achieved by Algorithm 1.
\begin{theorem}
The stable matching $\Phi^*$ converged in BECMA is Pareto optimal for BL chunks.
\end{theorem}

\par
The proof can be completed similar to that in \cite{gale1962college}. Since BL chunks are allowed to make proposals in Algorithm 1, we can refer it as a BL-driven matching algorithm. If we allow EL chunks to propose, we can similarly form an EL-driven matching algorithm, which generally produces different performance. According to \cite{gale1962college}, proposing players would be better off than proposed players, and achieve Pareto optimality. However, in the proposed BECMA, we can conclude the following remark.
\begin{remark}
The BL-chunk optimal BL-driven matching algorithm can have quite similar distortion performance, compared with the EL-chunk optimal EL-driven matching algorithm.
\end{remark}

\par
This property as described above is due to the formulated preference list in Sec. \ref{subsubsec:preference} is {\em mutual}, since the preferences of BL chunks and EL chunks are consistent in terms of sum-distortion. In Sec. \ref{subsec:performance-comparison}, we will conduct performance comparison between these two algorithms and validate our conclusion.
\section{Performance Evaluation}\label{sec:evaluation}
We carry out simulations to evaluate performance of the proposed Supcast under various scenarios. For a comprehensive evaluation over videos with different spatial-temporal content complexities, we take multiple standard reference video sequences from the Xiph \cite{Xiph} as our test sequences, including 'Bus', 'Coastguard', 'Crew', 'Foreman', 'Harbour', 'Husky' and 'Ice'. For fair comparison with Softcast, the luminance component of these video sequences is also extracted to generate monochrome versions in the simulation. For color versions, we can process chrominance components as the same as luminance components. Same as Softcast, the used video sequences are with common intermediate format (CIF) resolution and the frame rate of 30 fps (frame per second). In this case, the source bandwidth $BW_s=1.52$ MHz (in complex symbols). The GOP size is set as 4 and the default equal chunk division after 3D-DCT is set as $8\times 8$ chunks per frame. The peak signal-to-noise ratio (PSNR) is adopted as the performance metric, which can be expressed as
\begin{equation}
{\rm PSNR} = 10\log_{10}(255^2/{\rm MSE}),
\end{equation}
where the distortion MSE is averaged over all pixels in a frame. Besides, the PSNR of each video sequence is the average PSNR of frames.
\par

\par
In simulations, a single-cell NOMA system is considered with one BS located in the cell center, and ten heterogeneous multicast users equally partitioned into two distinct zones. Five near users are deployed within a ring between radius 100 meters and 500 meters, while another five far users are deployed within a ring between radius 500 meters and 900 meters. In each zone, all users are uniformly and randomly distributed. The average power budget for transmitting each chunk is set as $P=1$ W. The path-loss exponent is set as $\eta=2$. As discussed in Sec. \ref{subsec:video-reconstruction}, the AWGN of the fading channel for each user has the variance $\sigma_w^2$. Hence, the average channel signal-to-noise ration (SNR) is defined as $10\log_{10}(P/\sigma_w^2)$. I-Q modulation is adopted in the physical layer, which means each channel use can convey one symbol as I component and one symbol as Q component. These two symbols compose a complex symbol. We define the bandwidth compression ratio as $\beta=BW_c/BW_s$, where $BW_c$ is the channel bandwidth (in complex symbols), i.e., the number of available channel use. Without specific instructions, the default setting of $\beta$ is 0.5.

\subsection{Performance Comparison}\label{subsec:performance-comparison}
We first compare performance of the proposed Supcast against two reference schemes. One is Softcast developed in OMA which discards chunks with smallest variances in the case of insufficient bandwidth. This implies that half chunks will not be transmitted when bandwidth compression ratio $\beta=0.5$. Another scheme is referred as NOMA-RA, where BL chunks and EL chunks are randomly scheduled for superposition in NOMA. In this scheme, power of each chunk is allocated in the way the same as in Softcast.
\begin{figure}[htb]
\vspace{-0.35cm}
\centering
\includegraphics[width=6.5cm]{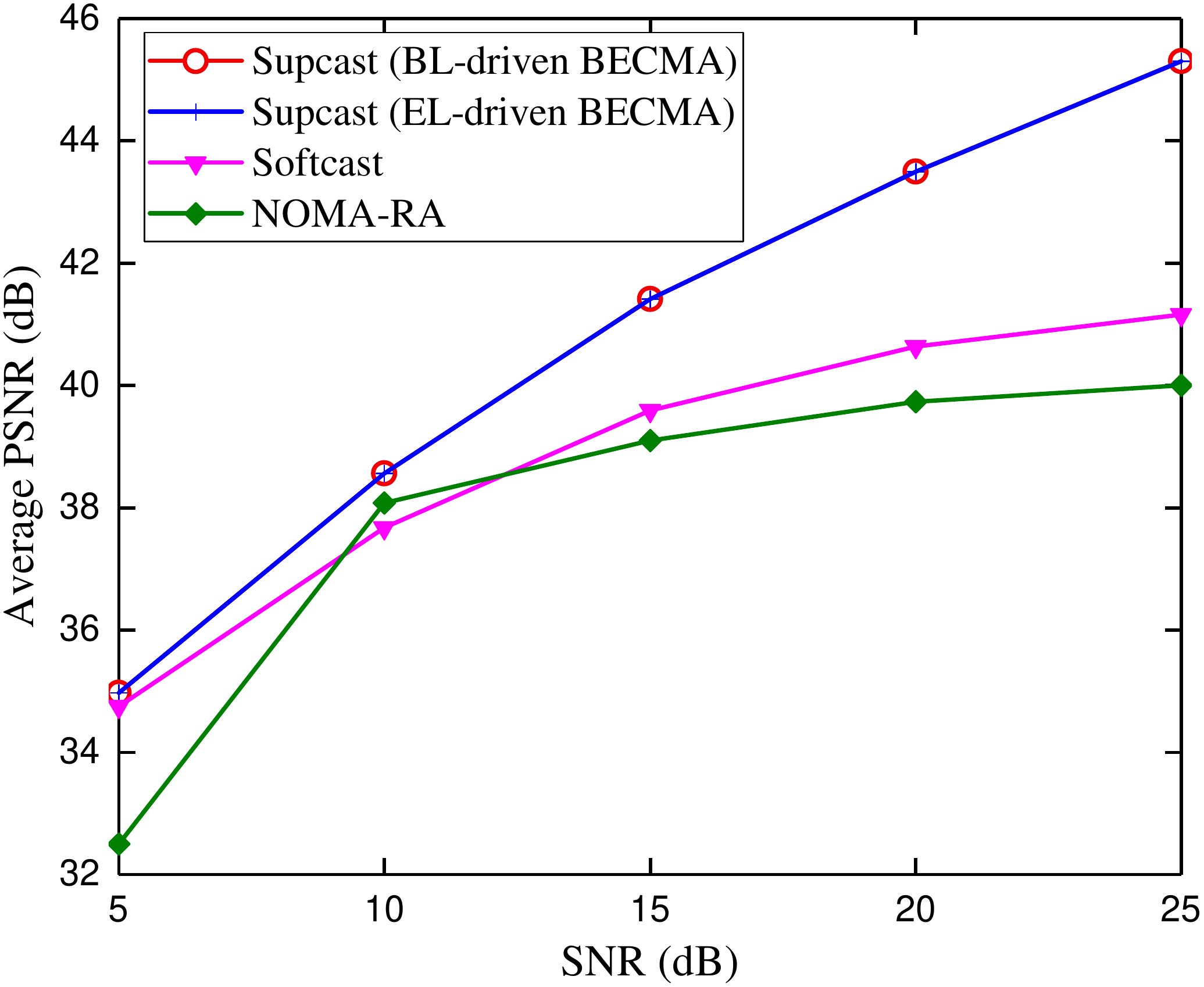}
\vspace{-0.2cm}
\caption{Average PSNR performance comparison among different schemes, $\beta=0.5$, GOP $=4$, 64 chunks/frame.}\label{figure:system-performance}
\vspace{-0.2cm}
\end{figure}
\par
As mentioned in Sec. \ref{subsec:algorithm-analysis}, the BECMA expressed in Algorithm 1 is actually a BL-driven chunk scheduling algorithm, where BL chunks make the proposal. By setting EL chunks as the proposer, we can generate an EL-driven algorithm. In simulations, we implement Supcast with both two BECMA algorithms. For these schemes, system performance is evaluated by averaging PSNR over all test sequences, as Fig. \ref{figure:system-performance} shows. Note that performance of Supcast with BL-driven BECMA is similar to that with EL-driven BECMA. This result is consistent with the conclusion in Remark 1, since formulated preferences of BL chunks and EL chunks are consistent in terms of sum-distortion.

\par
As shown in Fig. \ref{figure:system-performance}, Supcast outperforms Softcast and NOMA-RA over the entire range of SNR, and achieves gains up to 3 dB in terms of average PSNR. When SNR is low at 5 dB, performance of Supcast is similar with Softcast. This is due to the fact that the poor channel condition can only support transmission of all BL chunks and very few EL chunks with SC. In this case, most EL chunks are allocated with no power in the power re-allocation stage according to Eq. (\ref{eq:power-reallocation-solution}). In contrast, without two-stage power allocation, NOMA-RA has poor performance, since reconstruction performance of BL chunks degrades with severe interference of randomly superposed EL chunks. At medium and high SNR, performance gains of Supcast enlarge accordingly, while both NOMA-RA and Softcast start the effect of performance saturation. However, saturation reasons of two schemes are totally different. The former is due to severe interference of SC cannot be well alleviated in NOMA-RA with random scheduling and simple power allocation strategy. The latter is because distortion of discarding half chunks in Softcast is nonrecoverable.

\begin{table}[htb]
\centering
\caption{PSNR of each video sequence under different schemes.}\label{tab:per-video}
\begin{tabular}{|c|C{0.7cm}|C{0.7cm}|C{0.7cm}|C{0.6cm}|C{0.6cm}|C{0.6cm}|}
     \hline

     \multirow{2}*{{\diagbox{Video}{SNR}}} & \multicolumn{2}{c|}{Supcast} & \multicolumn{2}{c|}{Softcast} & \multicolumn{2}{c|}{NOMA-RA}\\
     \cline{2-7}  & 15dB & 25dB & 15dB & 25dB & 15dB & 25dB\\
     \hline
     Bus          & 38.77 & 42.85  & 37.03 & 38.62 & 37.08 & 37.89\\
     \hline
     Coastguard   & 43.68 & 47.72  & 42.05 & 43.69 & 41.89 & 42.69\\
     \hline
     Crew         & 45.96 & 50.26  & 44.54 & 46.38 & 43.67 & 44.59\\
     \hline
     Foreman      & 43.77 & 47.80  & 42.13 & 43.82 & 41.17 & 41.99\\
     \hline
     Harbour      & 41.80 & 45.88  & 40.28 & 42.00 & 40.01 & 40.76\\
     \hline
     Husky        & 32.26 & 35.48  & 28.88 & 29.46 & 28.55 & 29.04\\
     \hline
     Ice          & 43.62 & 47.78  & 42.21 & 44.12 & 41.9 & 42.79\\
     \hline
     \bf{ Average} & \bf 41.41 & \bf 45.39 & \bf 39.59 & \bf 41.16 & \bf 39.18 & \bf 39.96\\
     \hline
\end{tabular}
\end{table}

\par
Table \ref{tab:per-video} provides PSNR performance of each video sequence under different schemes. Due to space limits, we only show results when SNR is 15 dB and 25 dB. The results show that Supcast can achieve performance gains over all test sequences. For the integrity of evaluation, we also compare the proposed BECMA scheduling with the optimal user scheduling through exhaustive search, as Fig. \ref{figure:exhaustive} shows. Since the computational complexity of the optimal exhaustive search is $\mathcal{O}(M!)$, we perform the simulation by setting the GOP size as 1 and dividing coefficients into 16 chunks. The results in Fig. \ref{figure:exhaustive} illustrate that the performance of the BECMA approaches the upper bound generated from exhaustive search.
\begin{figure}[htb]
\vspace{-0.2cm}
\centering
\includegraphics[width=5.5cm]{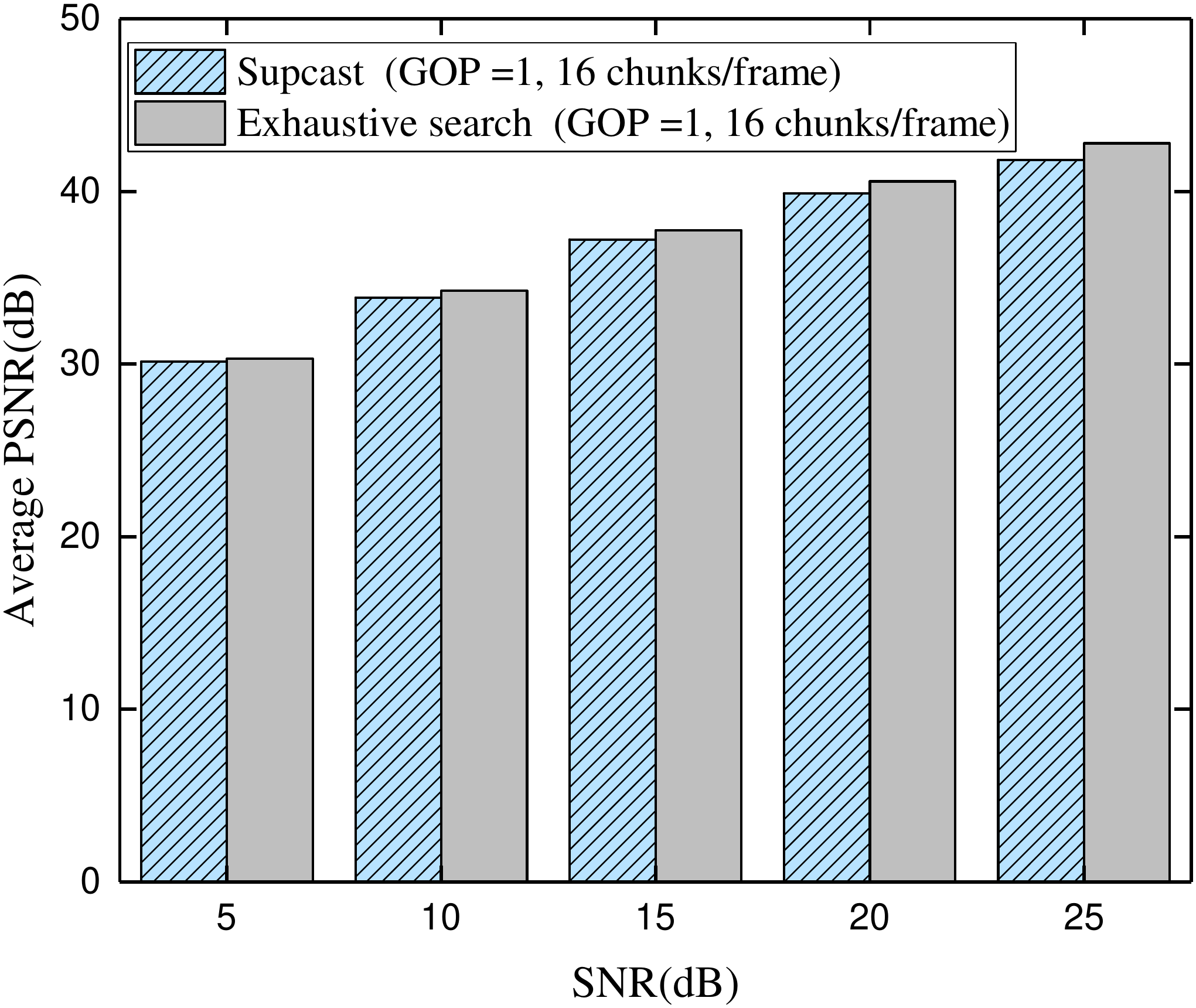}
\vspace{-0.2cm}
\caption{PSNR comparison between the BECMA scheduling and the optimal exhaustive search scheduling, $\beta=0.5$, GOP $=1$, 16 chunks/frame.}\label{figure:exhaustive}
\vspace{-0.2cm}
\end{figure}

\subsection{Impacts of Bandwidth Compression Ratio $\beta$}\label{subsec:compression-ratio}
Although Supcast has the ability to delivery twice chunks than Softcast owing to SC in NOMA, it has to discard some least important chunks before bisecting BL and EL chunks when bandwidth is severely insufficient. This leads to varying variance disparities between BL chunks and EL chunks in different bandwidth compression cases, which would accordingly affect the results of power allocation and chunk scheduling. Hence, we conduct simulations under different values of $\beta$, as illustrated in Fig. \ref{figure:peruser-performance}.

\par
From Fig. \ref{figure:peruser-performance} we can observe that Supcast at $\beta=0.5$ shows graceful degradation with varying SNR settings, while performance of Supcast at $\beta=0.25$ and performance of Softcast begin saturating at high SNR. This can be explained as follows. At $\beta=0.5$, all chunks can be transmitted and decoded by near users in Supcast. However, half chunks would be dropped in Softcast when $\beta=0.5$ and in Supcast when $\beta=0.25$, especially only quarter chunks can be transmitted in Softcast at $\beta=0.25$. Distortion of discarding information cannot be recovered at receivers. Thus, performance of these schemes is bounded at better channel conditions.

\begin{figure}[htb]
\vspace{-0.2cm}
\centering
\includegraphics[width=5.5cm]{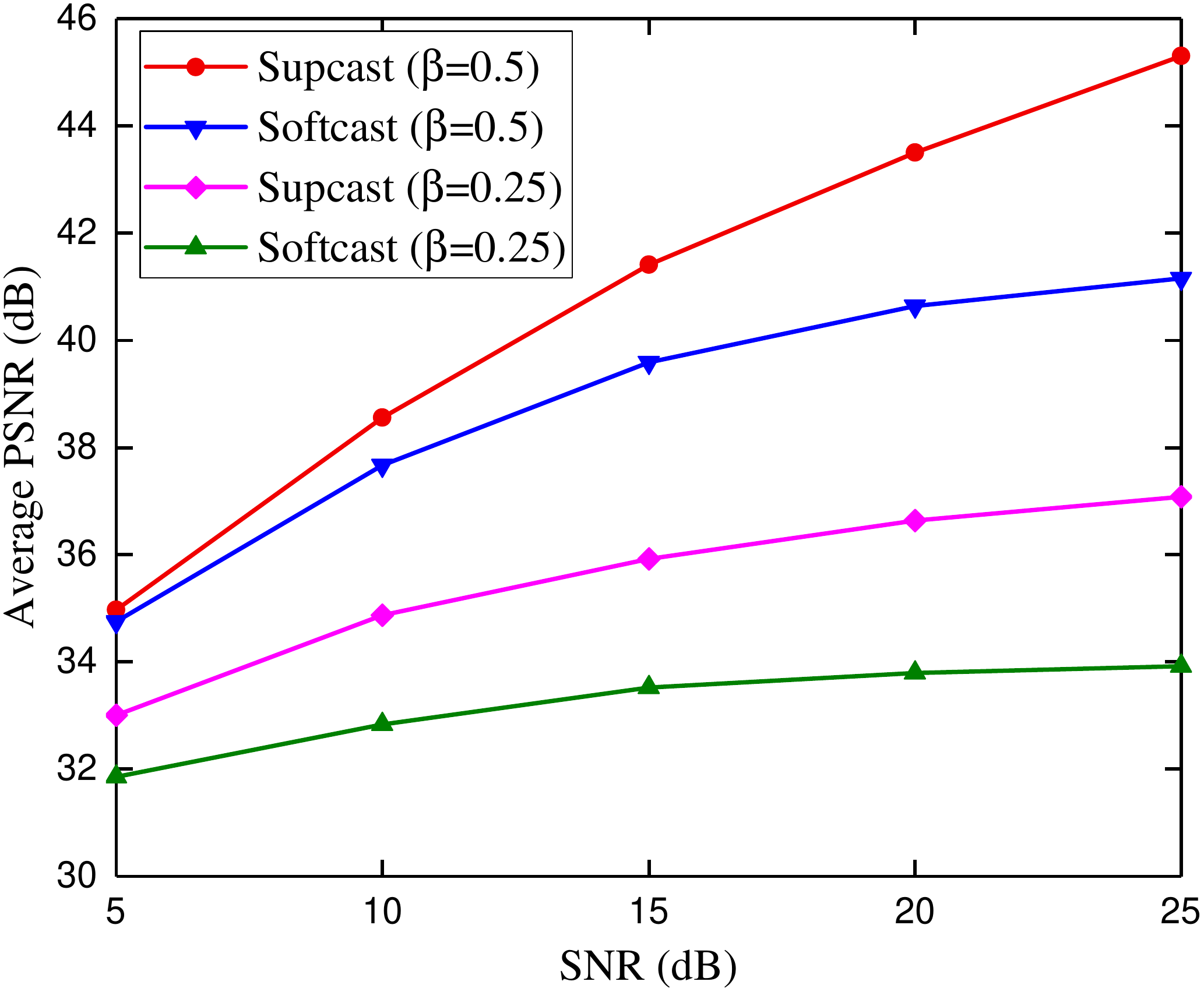}
\vspace{-0.2cm}
\caption{PSNR performance at different settings of bandwidth compression ratio $\beta$, GOP $=4$, 64 chunks/frame.}\label{figure:peruser-performance}
\vspace{-0.2cm}
\end{figure}

\par
Another interesting observation is compared with Softcast, Supcast can obtain higher performance gains at $\beta=0.25$ than $\beta=0.5$, when channel SNR is low. The reason can be expressed as follows. At $\beta=0.5$, poor channel condition can only support transmission of BL chunks and few superposed EL chunks. Thus the advantage of NOMA cannot be fully utilized, which induces   Supcast and Softcast have similar performance. However, at $\beta=0.25$, power is intensively allocated among transmitted chunks by sacrificing the transmission opportunity of some least important chunks. In this case, more EL chunks can be superposed and decoded by near users, even when channel condition is poor. This brings about remarkable gains for Supcast.

\subsection{Impacts of Chunk Size}\label{subsec:chunk-size}

We also investigate impacts of chunk size on performance of the proposed Supcast, since power allocation and scheduling are both carried out over chunks. DCT coefficients of each frame are divided into $N_c\times N_c$ chunks. In the simulation, we set $N_c$ as 4, 8 and 16, which produce a total number of chunks per frame as 16, 64 and 256, respectively. The results are depicted in Fig. \ref{figure:chunk_size}. Obviously, performance improves gradually with the increasing number of divided chunks. This is due to the fine-grained power allocation and chunk scheduling, which can accordingly produce more power scaling gains and superposition gains. Besides, the {\em marginal effect} would occur, which means the performance gap between the division of 256 chunks and 64 chunks is far less than the gap between the division of 64 chunks and 16 chunks.
\begin{figure}[htb]
\centering
\includegraphics[width=5.5cm]{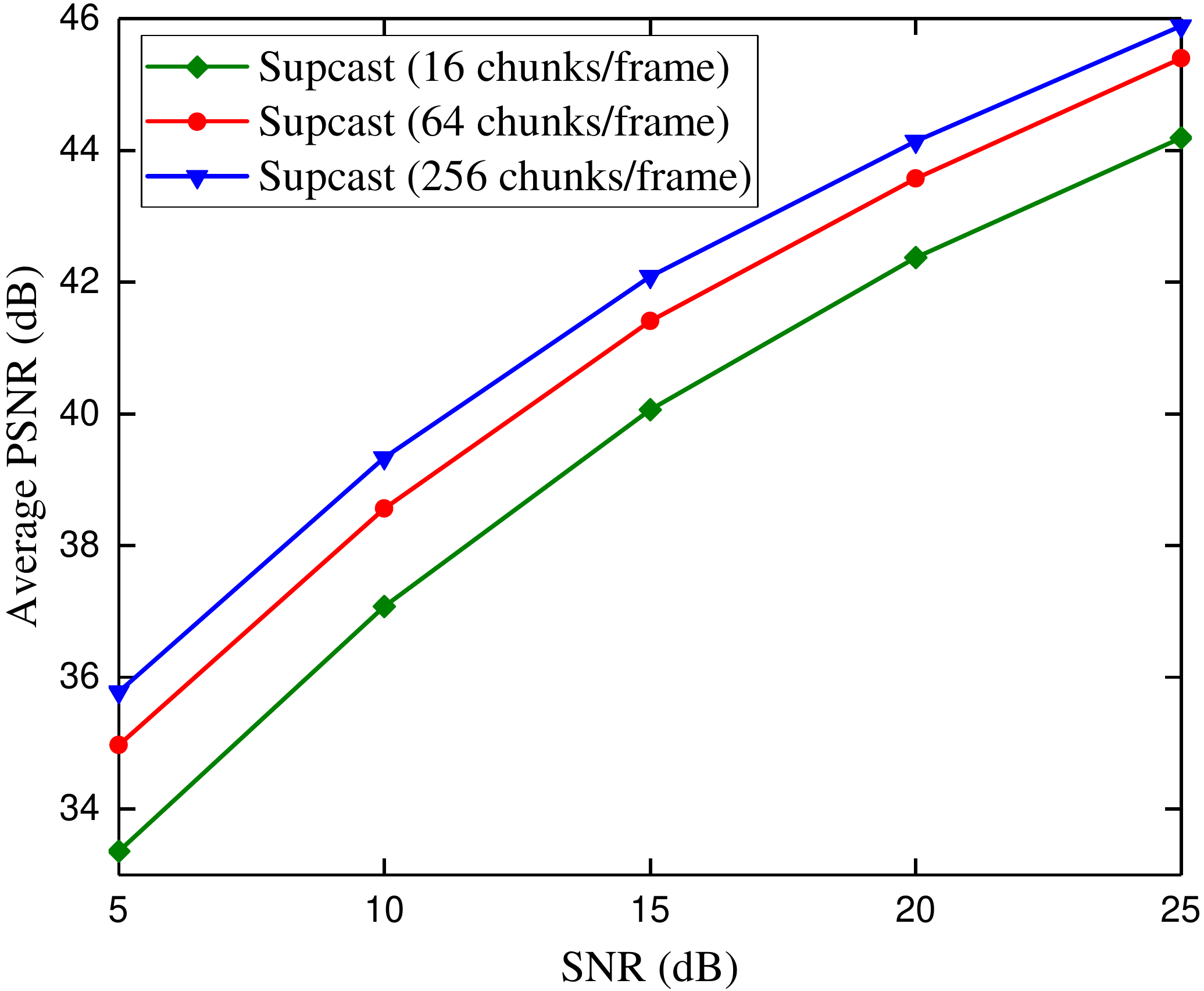}
\vspace{-0.2cm}
\caption{PSNR performance of Supcast under various settings of chunk division, $\beta=0.5$, GOP $=4$.}\label{figure:chunk_size}
\vspace{-0.2cm}
\end{figure}

\par
The above observations can provide guidelines for chunk division in Supcast, which implies the division of 64 chunks per frame can strike a good balance between performance and complexity. Too fine division would cause heavy computational burden for chunk scheduling, according to the analysis in Sec. \ref{subsec:algorithm-analysis}. Moreover, the huge overhead for transmitting meta data, including power scaling vectors and chunk scheduling vectors, is undesired when the number of chunks is too large. In this case, improved performance is not worthy due to the marginal effect. It is worth noting that the choice of appropriate chunk division should be adaptively adjusted according to the GOP size and the resolution of video sequences.

\section{Conclusions}\label{sec:conclusion}
In this paper, we presented Supcast, a novel receiver-driven scheme for video multicast in NOMA systems with heterogeneous channel conditions. The design of Supcast is motivated jointly by high spectral efficiency achieved with NOAM and robustness transmission of Softcast in heterogeneous environments. By grouping DCT coefficients into BL chunks and EL chunks for superposition, Supcast enables near users to decode both BL and EL chunks while far users to decode only BL chunks. Furthermore, inspired by Softcast, Supcast allows each heterogeneous user to obtain video with quality proportional to its channel quality. The key features of Supcast include power allocation and chunk scheduling to minimize video transmission distortion. Based on characteristics of chunks and interference formulation in decoding, we proposed a two-stage power allocation strategy. By reformulating chunk scheduling problem as a one-to-one two-sided matching problem, a near-optimal and low-complexity BECMA scheduling algorithm is proposed. Simulation results have shown remarkable performance improvements of Supcast over existing schemes.

\section*{Acknowledgements}
This work was supported in part by the National Science Foundation of China (NSFC) (Grants 91538203, 61390513, 61631017 and 61771445),
and the Fundamental Research Funds for the Central Universities.

\ifCLASSOPTIONcaptionsoff
  \newpage
\fi


\footnotesize
\begin{thebibliography}{99}
\bibitem{dai2018survey}
L. Dai, B. Wang, Z. Ding, Z. Wang, S. Chen and L. Hanzo, ``A  survey of non-orthogonal multiple access for 5G,'' {\em Commun. Surveys Tuts.}, vol. PP, no. 99, pp. 1-1, 2018.

\bibitem{ding2017survey}
Z. Ding, X. Lei, G. K. Karagiannidis, R. Schober, J. Yuan, and V. K. Bhargava, ``A survey on non-orthogonal multiple access for 5G networks: research challenges and future trends.'' {\em IEEE J. Sel. Areas Commun.}, vol. 35, no. 10, pp. 2181--2195, Oct. 2017.

\bibitem{islam2017power}
S. M. R. Islam, N. Avazov, O. A. Dobre, and K. S. Kwak, ``Power-domain non-orthogonal multiple access (NOMA) in 5G systems: potentials and challenges,'' {\em Commun. Surveys Tuts.}, vol. 19, no. 2, pp. 721--742, May. 2017.

\bibitem{index2016forecast}
Cisco, Visual Networking Index, ``Forecast and methodology, 2016-2021,'' June. 2017.

\bibitem{jiang2018enabling}
X. Jiang, H. Lu, and C. W. Chen, ``Enabling quality-driven scalable video transmission over multi-user NOMA system,''  in {\em Proc. IEEE INFOCOM}, Apr. 2018. pp. 1--9.

\bibitem{Choi2015minimum}
J. Choi, ``Minimum power multicast beamforming with superposition coding for multiresolution broadcast and application to NOMA systems,'' {\em IEEE Trans. Commun.}, vol. 63, no. 3, pp. 791-800, Mar. 2015.

\bibitem{zhang2018tradeoffs}
Z. Zhang, Z. Ma, M. Xiao, X. Lei, Z. Ding and P. Fan ``Tradeoffs of non-orthogonal multicast, multicast, and unicast in ultra-dense networks,'' {\em IEEE Trans. Commun.}, vol. PP, no. 99, pp. 1-1, 2018.


\bibitem{DVB}
{\em Digital Video Broadcasting (DVB)}. \url{http://www.etsi.org/deliver/etsi_en/300700_300799/300744/01.06.01_60/en_300744v010601p.pdf}.


\bibitem{softcast}
S. Jakubczak and D. Katabi, ``A cross-layer design for scalable mobile video,'' {\em in Proc. ACM Mobicom}, Sep. 2011. pp. 289--300.

\bibitem{parcast}
X. L. Liu, W. Hu, Q. Pu, F. Wu, and Y. Zhang, ``ParCast: soft video delivery in MIMO-OFDM WLANs,'' {\em in Proc. Mobicom}, Sep. 2012, pp. 233--244.

\bibitem{cui2014robust}
H. Cui, C. Luo, C. W. Chen, and F. Wu, ``Robust uncoded video transmission over wireless fast fading channel,'' in {\em Proc. IEEE INFOCOM}, Apr. 2014. pp. 1--9.

\bibitem{liu2018cg}
H. Liu, R. Xiong, X. Fan, D. Zhao, Y. Zhang and W. Gao, ``CG-Cast: scalable wireless image SoftCast using compressive gradient,'' {\em IEEE Trans. Circuits Syst. Video Technol.},  vol. PP, no. 99, pp. 1-1, 2018.

\bibitem{tan2017analog}
B. Tan, J. Wu, Y. Li, H. Cui, W. Yu, and C. W. Chen, ``Analog coded SoftCast: a network slice design for multimedia broadcast/multicast,''
{\em IEEE Trans. Multimedia}, vol. 19, no. 10, pp. 2293--2306, Oct. 2017.

\bibitem{he2017progressive}
D. He, C. Lan, C. Luo, E. Chen, F. Wu and W. Zeng, ``Progressive pseudo-analog transmission for mobile video streaming,'' {\em IEEE Trans. Multimedia}, vol. 19, no. 8, pp. 1894--1907, Aug. 2017.


\bibitem{liang2017hybrid}
F. Liang, C. Luo, R. Xiong, W. Zeng and F. Wu, ``Hybrid digital-analog video delivery with Shannon-Kotel¡¯nikov mappings'', {\em IEEE Trans. Multimedia}, vol. PP, no. 99, pp. 1--1, 2017.

\bibitem{elbamby}
M.S. Elbamby, M. Bennis, W. Saad, M. Debbah, and M. Latva-aho, ``Resource optimization and power allocation in in-band full duplex (IBFD)-enabled non-orthogonal multiple access networks,'' {\em IEEE J. Sel. Areas Commun.}, vol. 35, no. 12, pp. 2860--2873, July. 2017.

\bibitem{cui2018optimal}
J. Cui, Y. Liu, Z. Ding, P. Fan, and A. Nallanathan, ``Optimal user scheduling and power allocation for millimeter wave NOMA systems,'' {\em IEEE Trans. Wireless Commun.}, vol. 17, no. 3, pp. 1502--1517, Mar. 2018.

\bibitem{yang2018power}
Z. Yang, J. A. Hussein, P. Xu, Z. Ding and Y. Wu, ``Power allocation study for non-orthogonal multiple access networks with multicast-unicast transmission,'' {\em IEEE Trans. Wireless Commun.}, vol. PP, no. 99, pp. 1--1, 2018.

\bibitem{zhai2018energy}
D. Zhai, R. Zhang, L. Cai, B. Li, and Y. Jiang, ``Energy-efficient user scheduling and power allocation for NOMA based wireless networks with massive IoT devices,'' {\em IEEE Internet of Things J.}, vol. 5, no. 3, pp. 1857--1868, June. 2018.

\bibitem{wu2018optimal}
Y. Wu, L. P. Qian, H. Mao, X. Yang, H. Zhou and X. S. Shen, ``Optimal power allocation and scheduling for non-orthogonal multiple access relay-assisted networks,'' {\em IEEE Trans. Mobile Comput.}, vol. PP, no. 99, pp. 1--1, 2018.

\bibitem{fang2017joint}
F. Fang, H. Zhang, J. Cheng, S. Roy, and V. C. M. Leung, ``Joint user scheduling and power allocation optimization for energy-efficient NOMA systems with imperfect CSI,'' {\em IEEE J. Sel. Areas Commun.}, vol. 35, no. 12, pp. 2874--2885, Dec. 2017.

\bibitem{di2016joint}
B. Di, Siavash Bayat, Lingyang Song, Yonghui Li and Zhu Han, ``Joint user pairing, subchannel, and power allocation in fullduplex multi-user OFDMA networks,'' {\em IEEE Trans. Wireless Commun.},
vol. 15, no. 12, pp. 8260-8272, Dec. 2016.

\bibitem{boyd2004convex}
S. P. Boyd and L. Vandenberghe, {\em Convex Optimization}. Cambridge, U.K.: Cambridge Univ. Press, 2004.

\bibitem{shen2007global}
P. P. Shen and G. X. Yuan, ``Global optimization for the sum of generalized polynomial fractional functions,'' {\em Math. Meth. Oper. Res.}, vol. 65, no. 3, pp. 445--459. Jun. 2007.

\bibitem{baron2011peer}
E. Bodine-Baron, C. Lee, A. Chong, B. Hassibi, and A. Wierman, ``Peer effects and stability in matching markets,'' {\em in Proc. SAGT}, Oct. 2011, pp. 117--129.

\bibitem{roth1992two}
A. E. Roth and M. A. O. Sotomayor, {\em Two-Sided Matching: A Study in Game-Theoretic Modeling and Analysis}. Cambridge, U.K.: Cambridge Univ. Press, 1992.

\bibitem{bayat2014distributed}
S. Bayat, R. H. Y. Louie, Z. Han, B. Vucetic, and Y. Li, ``Distributed user association and femtocell allocation in heterogeneous wireless networks,'' {\em IEEE Trans. Commun.}, vol. 62, no. 8, pp. 3027--3043, Aug. 2014.

\bibitem{gale1962college}
D. Gale and L. S. Shapley, ``College admissions and the stability of marriage,'' {\em Amer. Math. Monthly}, vol. 69, no. 1, pp. 9--15, Jan. 1962.

\bibitem{Xiph}
Xiph.org media. \url{http://media.xiph.org/video/derf/}.
\end{thebibliography}
\end{document}